\documentclass[10pt, twocolumn]{IEEEtran}

\usepackage{amsmath}
\usepackage{amssymb}
\usepackage{amsfonts}
\usepackage[english]{babel}
\usepackage{cite} 

\usepackage{caption} 
\usepackage{graphicx}
\usepackage{epsfig}
\usepackage{subfigure}
\usepackage{tablefootnote}

\newtheorem{theorem}{Theorem}
\newtheorem{remark}{Remark}

\newtheorem{definition}{Definition}

\usepackage{accents}
\makeatletter
\def\widebar{\accentset{{\cc@style\underline{\mskip10mu}}}}
\def\Widebar{\accentset{{\cc@style\underline{\mskip13mu}}}}
\makeatother

\begin{document}

\captionsetup[figure]{labelfont={ }, name={Fig.}, labelsep=period} 
\pagestyle{empty}

\title{ Capacity Region of Gaussian Multiple-Access Channels with Energy Harvesting\\ and Energy Cooperation }

\author{Yunquan~Dong,~\IEEEmembership{Member,~IEEE},
        Zhengchuan Chen,~\IEEEmembership{Member,~IEEE},
        Pingyi~Fan,~\IEEEmembership{Senior Member,~IEEE}\\
\thanks{
Y. Dong is with the School of Electronic and Information Engineering,  Nanjing University of Information Science and Technology, Nanjing 210044, China (e-mail: yunquandong@nuist.edu.cn).

Z. Chen is with Information Systems Technology and Design Pillar, SUTD, Singapore (email: zhengchuan\_chen@sutd.edu.sg).

P. Fan is with Department of Electrical Engineering, Tsinghua University, Beijing 100084, China.  
Technology, Tsinghua University, Beijing 100084,  China (email: fpy@tsinghua.edu.cn).
%
}
}

\maketitle

\begin{abstract}
    We consider the capacity region of a $K$-user multiple access channel (MAC) with energy harvesting transmitters.
        Each user stores and schedules the randomly arriving energy using an energy buffer.
    Users can also perform energy cooperation by transmitting energy to other users or receiving energy from them.
        We derive the capacity region of this channel and show that 1) the capacity region coincides with that of a traditional $K$-user Gaussian MAC with energy cooperation, where the average power constraints are equal to the battery recharging rates of the energy harvesting case;
    2) each rate on the capacity region boundary can be achieved using the save-and-forward  power control and a fixed energy cooperation policy.
\end{abstract}

\begin{keywords}
   Multiple-access channel, capacity region, energy harvesting, energy cooperation.
\end{keywords}

\section{Introduction}
Internet of Things (IoT) has been studied extensively by both industry and academia  in recent years.
    In IoT networks, the Internet is connected to the physical world via ubiquitous wireless sensor networks (WSNs), which consists of a wide range of massively-deployed sensing devices.
Despite the wide applications of IoT networks, their performance  is  severely constrained by the capacity of sensor batteries.
        To address this issue, energy harvesting WSNs  (EH-WSNs) and the energy cooperation technology have been developed and widely researched \cite{WSN-eh-2011, Coop_ISIT_2012}.
     In EH-WSNs, each node can harvest energy (e.g., solar and wind power) from the ambient environment by employing an energy harvesting unit and an energy buffer.
This provides each node with almost a perpetual energy supply, and thus extends the network lifetime significantly.
    By employing an energy transceiver, each node can also transmit some energy to other nodes in one time slot and receive energy  from others in another time slot, so that the utilization of  the available energy over the network could be optimized, referred to as \textit{energy cooperation} \cite{Coop_ISIT_2012, Coop_Asilomar_2012, Coop_Tcom_2013, Zhi-2016}.

This paper investigates the capacity region of a $K$-user Gaussian multiple-access channel (MAC) using energy harvesting and energy cooperation, which corresponds to the uplink communication  of EH-WSNs.
    We aim at characterizing the performance limit of this channel and the capacity achieving power control/energy cooperation protocols.
    We first consider the capacity region of a Gaussian MAC with energy cooperation and average power constraints (i.e., powered by traditional batteries) as a baseline performance.
We show that each point on the capacity region boundary is achievable using a \textit{fixed energy cooperation policy} and time-sharing among cooperation policies is not required.
       Second, we investigate the capacity region of a Gaussian MAC with energy cooperation and energy harvesting constraints (i.e., powered by energy harvesting).
We show that the capacity of a Gaussian MAC with energy cooperation and energy harvesting constraints is equal to that of a Gaussian MAC with energy cooperation and average power constraints.
    In particular, the capacity region can be achieved using a \textit{save-and-forward power control},  where each node saves all the harvested energy for a certain period first, and performs information transmission as well as energy cooperation afterwards.

\subsection{Related Works}
In energy harvesting powered systems, the energy harvesting process is random over time.
    The harvested energy also suffers from the causality constraint, i.e., nodes can only use the energy harvested in the past.
Thus, each sensor may suffer from occasional energy shortages.
   To reduce energy shortage events and improve energy efficiency, the harvested energy needs to be scheduled carefully.
For the point-to-point fading channel, it has been shown that the directional water-filling power allocation achieves the maximum throughput  \cite{DF_JSAC_2011}.
    For multiple access channels, the maximum departure region can be achieved by a generalized water-filling based power allocation \cite{MAC_JCN_2012}.
Moreover, the sum-rate optimal scheduling for energy-harvesting interference channels was developed in \cite{IC_JCN_2012}.

Energy cooperation allows users to share their harvested energy through wireless power transfer \cite{ Xiong-2015, Xiong-2016}.
    By transferring some energy to nodes with better channel conditions, energy cooperation can enhance  network performance significantly.
For example, by scheduling energy among users over time, the two-dimensional directional water-filling scheme achieves the optimal throughput of both two-way channels and two-user multiple access channels \cite{EC_TWC_2015,EC_ITA_2013,Kaya_ITA_2016};
    by transferring energy to source nodes, the digital network coding with energy cooperation even outperforms the physical network coding \cite{Zhi-2016} in two-way relay networks.
    Moreover, nodes can also perform energy cooperation by using cooperative relays or mobile control centers \cite{Litao-2016,Litao-2016-2}.
Recently, \cite{Gurakan-2016} investigated the interesting interplay between data cooperation and energy cooperation in multi-access channels, which sheds some light on how to perform cooperation efficiently.

In addition to exploring energy/packet scheduling schemes and energy cooperation policies, the performance limits of  energy harvesting powered communications have also been studied  \cite{Ulukus-2012-awgn, DongIn-2015-ICeh, Yangjing-2014-awgn}.
    In \cite{Ulukus-2012-awgn}, the authors investigated information transmission over Gaussian channels using energy-harvesting transmitters.
By employing either the save-and-forward or the best-effort transmission scheme, it has been shown that the capacity of a Gaussian channel with average power constraint is also achievable by energy harvesting transmitters.
    Also, the asymptotic equivalence between energy harvesting powered sensing and traditionally powered sensing was shown in \cite{Yangjing-2014-awgn}.
	In addition, the achievable average rate region of the symmetric two-user interference channel with energy harvesting and energy cooperation was presented in \cite{DongIn-2015-ICeh}.

\textit{Notation}: Boldface letters indicate vectors and $(\cdot)^{\textsf{T}}$ denotes the transpose operation.
$\Phi=\{1,2,\cdots, K\}$ is a set of integers.
For an $n$-dimensional vector $\boldsymbol{x}=\{x_1,x_2\cdots,x_n\}\in \mathcal{R}^n$ and a subset $S\subseteq \Phi $, $\boldsymbol{x}(S)$ denotes $\sum_{k\in S}x_k$.

\section{System Model}\label{sec:2}
\begin{figure}[!ht]
\centering
\includegraphics[width=2.8in]{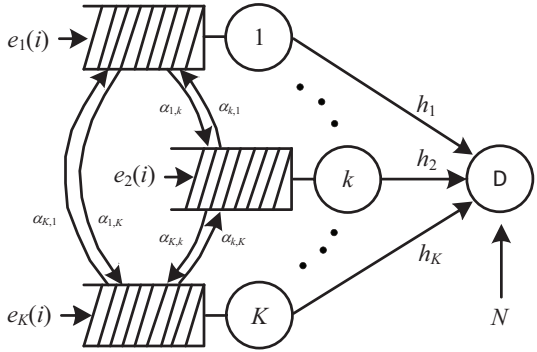}
\caption{The $K$-user Gaussian MAC with energy cooperation and energy harvesting transmitters.} \label{fig:mac model}
\end{figure}

Consider a $K$-user Gaussian MAC as shown in Fig. \ref{fig:mac model}, where
   each user sends its own message to the receiver independently.
Assume that each user has an energy harvesting unit so that it can  collect energy (e.g., solar energy and wind energy) from the environment.
   We also assume that each user is equipped with an energy transmitting/receiving unit so that it can transmit some of the harvested energy to other users, as well as receive the energy transmitted by other users.

\subsection{Energy Harvesting Model} \label{sbusec:2_1}

Suppose  time is slotted and  slot length is $T$.
In slot $i$, let $x_k(i)$ be the transmitted signal of user $k$ and $y(i)$ be the received signal at the receiver. We have
\begin{equation}\label{eq:mac_model}
  y(i)=\sum_{k=1}^K h_k x_{k}(i) + z(i),
\end{equation}
 where $h_k$ is the channel gain between user $k$ and the receiver, $z(i)$ is the Gaussian noise with zero-mean and variance $\sigma^2$.
We assume that $h_k$ is constant throughout each period of information transmission.

Let $e_k(i)$ denote the amount of energy that user $k$ harvests in slot $i$ and $\textbf{E}=[ e_k(i) ]_{K\times N}$ denote the energy harvesting matrix of all users over $N$ slots.
     Denote $\bar{e}_k=\mathbb {E} [e_k(i)]$ as the \textit{energy harvesting rate} (equals to the battery recharging rate) and $\widebar{\boldsymbol{e}}=[\bar{e}_1,\cdots,\bar{e}_K]^{\textsf{T}}$ as the energy harvesting rate vector.

To investigate the performance limit of the channel, we consider a subset of the following assumptions.
\begin{itemize}
  \item [A1] The energy buffer at each user is infinitely large.
  \item [A2] $\{e_k(i), i=1,2,\cdots\}$ is an ergodic, independent and identically distributed sequence.
  \item [A3] For each user, the expectation of the harvested energy in a slot is finite, i.e., $\bar{e}_k<\infty$.
\end{itemize}

Assumptions A1--A3 are used throughout the paper.
   By assuming energy buffer to be infinite, energy overflow is avoided. In fact, the capacity of a small button battery is more than 200 milliampere hour (mAh), which is large enough for most energy harvesting scenarios~\cite{button-2015}.
Assumption A2 implies that $\lim_{N\rightarrow\infty}\frac1N \sum e_k(i)=\mathbb{E}[e_k(i)]=\bar{e}_k$.

\subsection{Energy Cooperation Model} \label{subsec:2_2}

    In each slot, user $k$ may transmit some energy $\sum_{l=1}^K \delta_{kl}(i)$ to other users and receive  energy $\sum_{l=1}^K \alpha_{lk} \delta_{lk}$ from other users.
We denote $\alpha_{kj}\in(0,1)$ as the  energy transfer efficiency from user $k$ to user $j$\footnote
{
    When $K>2$, user $k$ can transfer energy to user $j$ by direct transmission or through relaying by some other users.
       For an energy path $k\rightarrow u_1\rightarrow \cdots\rightarrow u_m\rightarrow j$, the corresponding energy efficiency is $\alpha_{k,u_1}\alpha_{u_1,u_2}\cdots \alpha_{u_m,j}$.
    The most efficient energy route can be found using the Dijkstra's algorithm \cite{Dijkstra-1959}, by considering users as  vertices and $w_{kj}=-\ln\alpha_{kj}$ as the edge weight.
        In this paper, we are using $\alpha_{ij}$ to represent the resulting maximum energy transfer efficiency from user $k$ to user $j$.
}.
In slot $i$, if user $k$ transmits $\delta_{kj}(i)$ amount of energy to user $j$, then the received energy at user $j$ is $\alpha_{kj}\delta_{kj}(i)$.
   Since a node cannot transmit energy to itself, we denote $\alpha_{kk}=o$ and $\delta_{kk}(i)=0$,  where $o$ is an infinite small positive number.
In addition, we define the energy transfer matrix in slot $i$ as $\textbf{D}(i)=[\delta_{kj}(i)]_{K\times N}$.

We define the \textit{consumed power} of user $k$ in slot $i$ as
\begin{equation}\label{df:p_bar}
  \widetilde{p}_k(i)=p_k(i)+\frac{1}{T}\sum_{l=1}^K (\delta_{kl}(i)-\alpha_{lk}\delta_{lk}(i)),
\end{equation}
where $p_k(i)$ is the transmit power of user $k$ in slot $i$ and $T$ is  slot length.
    We call it consumed power because $T\widetilde{p}_k(i)$ is the total amount of energy depleted from user $k$'s energy buffer in the slot.
In addition, we denote $\widetilde{\boldsymbol{p}}_k(i)=[\widetilde{p}_1(i),\cdots,\widetilde{p}_K(i)]^{\textsf{T}}$ and $\widetilde{\textbf{P}}=[\widetilde{p}_k(i)]_{K\times N}$ as the consumed power vector and the consumed power matrix, respectively.

At the end of slot $i$, the remaining energy of user $k$ can be expressed as
\begin{equation}\label{cons1-non_neg}
    E_k^\textrm{r}(i)=\sum_{j=1}^i \left( e_k(j) - T \widetilde{p}_k(j) \right).
\end{equation}

Note that $\widetilde{p}_k(i)$ should be chosen such that the remaining energy $E_k^\textrm{r}(i)$ is non-negative.
   Thus, the consumed power must satisfy the following constraint (CSTR).

\hspace{-4mm} CSTR1~(\textit{Energy Causality Constraint}): In each slot, the consumed power $\widetilde{p}_k(i)$ of each user $k$ satisfies
                \begin{equation}\label{cstr:energy_casl}
                  \sum_{j=1}^i T \widetilde{p}_k(i) \leq \sum_{j=1}^i e_k(j).
                \end{equation}

From CSTR1, it is clear that the causality of energy arrivals imposes a restriction on the consumed power $\widetilde{p}_k(i)$, other than on the transmit power $p_k(i)$.
   Therefore, we shall investigate the energy cooperation policy and power control policy based on $\widetilde{p}_k(i)$ in the following sections.

\subsection{Optimal Policies} \label{subsec:2_3}
Under Assumptions A1--A3, the achievable rate region can be found by re-distributing energy among users and  scheduling energy over time.
   To be specific, an energy cooperation policy determines how much energy will be transferred to other users and a power control policy determines how much power should be allocated to each user in each slot.

\begin{definition}
A \textit{power control policy} $\mathcal{P}$ is a mapping from the energy harvesting matrix $\textbf{E}$ to the consumed power matrix $\widetilde{\textbf{P}}$.
    Given $\textbf{E}$, $\mathcal{P}_{k,i}(\textbf{E})$ is the allocated consumed power $\widetilde{p}_k(i)$ of user $k$ in slot $i$.
\end{definition}

 Note that a power control policy $\mathcal{P}$ is feasible only if the resulting $\widetilde{\textbf{P}}$ is positive and satisfies constraint CSTR1.
    The set of all feasible power control policy is denoted as
\begin{eqnarray}\label{df:p_fsble} \nonumber
  \mathcal{F}_P \hspace{-3 mm}&=&\hspace{-3 mm}\Bigg\{ \mathcal{P}:
                        \sum_{j=1}^i \widetilde{p}_k(i)T \leq \sum_{j=1}^i e_k(i),\widetilde{p}_k(i)\geq0,  \\
               & &~~~~~~~~~~~~~~~~~~
               1\leq k\leq K,1\leq i\leq N \Bigg\}.
\end{eqnarray}

\begin{definition}
An \textit{energy cooperation policy} $\mathcal{D}$ is a mapping from  the consumed power vector $\widetilde{\boldsymbol{p}}(i)$ to  the energy transfer matrix $\textbf{D}(i)$.
      Given $\widetilde{\boldsymbol{p}}(i)$, $\mathcal{D}_{kj}(\widetilde{\boldsymbol{p}}(i))$ can be interpreted as the transferred energy $\delta_{kj}(i)$ from user $k$ to user $j$.

According to Lemma 1 in \cite{EC_ITA_2013}, for any energy cooperation policy where some users transfer energy and receive energy at the same time, we can find an equivalent energy cooperation (resulting to the same transmit power vector) where each user either transfer energy or receive energy.
   Therefore, we do not need full-duplex energy transceivers and hence can only focus on energy cooperation policies satisfying $\delta_{jk}(i)\delta_{kl}(i)=0$. 
By the definition of $\widetilde{p}_k(i)$ (cf. (\ref{df:p_bar})) and the fact that both $p_k(i)$ and $\delta_{lk}(i)$ are non-negative, we thus define the set of all feasible energy cooperation policy as
\begin{eqnarray}\label{df:d_fslble}\nonumber
  \mathcal{F}_D\hspace{-3mm}&=&\hspace{-3mm}\Bigg\{ \mathcal{D}: \delta_{kj}(i) \geq 0, \sum_{j=1}^K \delta_{kj}(i)\leq T\widetilde{p}_k(i),  \\
                             &&~~~~~~~~~~~~~~~~~~
                             1\leq k\leq K, 1\leq j\leq K  \Bigg\}.
\end{eqnarray}
\end{definition}

\subsection{Normalized Channel Model} \label{subsec:2_4}
Without loss of generality, we perform analysis based on a \textit{normalized Gaussian MAC} with unit channel gains and unit noise power  \cite{EC_ITA_2013}.
   In particular, the signal-to-noise ratio (SNR) of each user in the normalized MAC is the same as that of the original channel, so that the capacity region of the normalized channel is also  equal to that of the original channel.
To be specific, the normalized Gaussian MAC is obtained by scaling the transmit power $p_k(i)$ of user $k$ with $h_k/\sigma^2$, scaling the harvested energy $e_k(i)$ with $h_k/\sigma^2$, and scaling the energy transmission efficiency $\alpha_{kj}$ with $h_j/h_k$.
    Note that  $\alpha_{kj}$ may be larger than one after the scaling.
   Nevertheless, this is only an artifact due to mathematical formulation and does not mean that the transferred energy will be amplified.

\begin{figure*}[htp]

\hspace{-6 mm}
    \begin{tabular}{cc}
    \subfigure[The capacity regions of 2-user Gaussian MAC.]
    {
    \begin{minipage}[t]{0.5\textwidth}
    \centering
    {\includegraphics[width = 88mm] {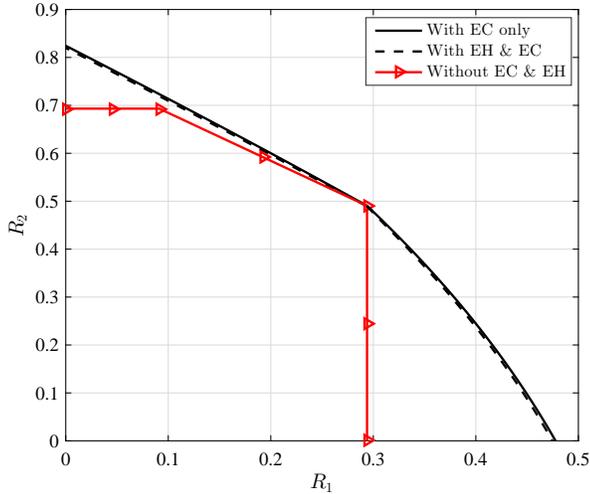} \label{fig:region}}
    \end{minipage}
    }

 \subfigure[The sum capacity of 2-user Gaussian MAC.]
    {
    \begin{minipage}[t]{0.5\textwidth}
    \centering
    {\includegraphics[width = 88mm] {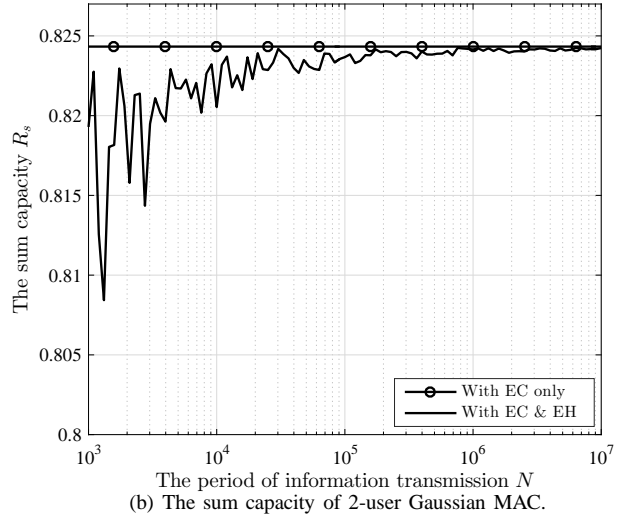} \label{fig:sum_c}}
    \end{minipage}
    }

    \end{tabular}

\caption{Approaching AWGN capacity region using energy harvesting.}
\end{figure*}

\section{The Capacity Regions} \label{sec:2}
In this section, we first investigate the capacity of Gaussian MACs with energy cooperation and average power constraint.
    Next, we study the capacity region of Gaussian MACs with energy cooperation as well as energy harvesting, and establish the equivalence between them.

\subsection{Gaussian MAC with Energy Cooperation}\label{subsec:2_1}

Given the transmit power $p_k$ of each user, the capacity region of a normalized Gaussian MAC with unit channel gain and unit noise power is known as \cite{Cover_1991}
\begin{equation}\label{cr:awgn}
  \mathcal{C}_\textrm{g}(\boldsymbol{p})\hspace{-1mm}=\hspace{-1mm}
    \left\{ \hspace{-0.5 mm} \boldsymbol{R}:\boldsymbol{R}(S) \hspace{-0.5 mm}
            \leq \frac12 \log\left( \hspace{-0.5 mm} 1+\sum_{k\in S}p_k \hspace{-0.5 mm}\right),
            ~\forall S\subseteq \Phi \hspace{-0.5 mm} \right\},
\end{equation}
where  $\boldsymbol{p}=[p_1,\cdots,p_K]^{\textsf{T}}$ is the transmit power vector, $\boldsymbol{R}=[ R_1,\cdots, R_K ]^{\textsf{T}}$ is the rate vector and $\Phi=\{1,2,\cdots,K\}$ is the set of users.
   The capacity region lies in the positive quadrant and has $K!$ vertices.

Given the consumed power $\widetilde{p}_k$ and an energy cooperation policy $\mathcal{D}$, the transmit power of user $k$ would be $p_k=\widetilde{p}_k-\frac{1}{T}\sum_{l=1}^K (\delta_{kl}-\alpha_{lk}\delta_{lk})$.
   Therefore, the capacity of the Gaussian MAC under a given energy cooperation policy is
\begin{eqnarray}\label{cr:ener-coop-D} \nonumber
 \hspace{-3mm} \mathcal{C}_{\textrm{g}}(\widetilde{\boldsymbol{p}},\mathcal{D}) \hspace{-3.2mm}&=&\hspace{-3.2mm}\Bigg\{ \boldsymbol{R}:\boldsymbol{R}(S) \leq
            \frac12   \log  \Bigg(  1 + \sum_{k\in S} \\
    & &  \left( \widetilde{p}_k  -  \frac{1}{T}
          \sum_{l=1}^K (\delta_{kl}-\alpha_{lk}\delta_{lk}) \right)  \Bigg) ,
  \forall S\subseteq \Phi \Bigg\}.
\end{eqnarray}

By considering all possible energy cooperation policies, the following theorem presents the capacity region of a Gaussian MAC with energy cooperation.

\begin{theorem}\label{th:cr-ener-coop}
   The capacity region of a normalized Gaussian MAC with energy cooperation is
   \begin{equation}\label{cr:ener-coop}
     \mathcal{C}_{\textrm{g,EC}}(\widetilde{\boldsymbol{p}})=\bigcup_{\mathcal{D}\in \mathcal{F}_D} \mathcal{C}_{\textrm{g}}(\widetilde{\boldsymbol{p}},\mathcal{D}),
   \end{equation}
   where $\mathcal{F}_D$ is the set of feasible energy cooperation policies.
\end{theorem}

\begin{proof}
    See Appendix \ref{app:prf-thm1}.
\end{proof}

This theorem demonstrates the improvement in capacity region due to the energy cooperation among users, which redistributes the available energy as desired.
   The above characterization also shows that each point on the capacity region boundary can be achieved by a single energy cooperation policy, and thus time sharing among energy cooperation policies is not required.

\subsection{Gaussian MAC with Energy Cooperation and Energy Harvesting}\label{subsec:2_2}

In the energy harvesting scenario, the harvested energy needs to be scheduled by controlling the transmit power of nodes.
   Therefore, the capacity region of a $K$-user Gaussian MAC with energy cooperation and energy harvesting is the set of achievable information rates under all possible power control policies and energy cooperation policies.

\begin{theorem}\label{th:capacity_eceh}
   The capacity region of a $K$-user Gaussian MAC with energy cooperation and energy harvesting is determined by the energy harvesting rate vector $\bar{\boldsymbol{e}}=[\bar{e}_1,\cdots,\bar{e}_K]^{\textsf{T}}$. In particular, we have
    \begin{equation}\label{rt:th_capacity_eceh}
        \mathcal{C}_\textrm{EC,EH}(\bar{\boldsymbol{e}}) = \mathcal{C}_\textrm{g,EC}(\widetilde{\boldsymbol{p}})|_{\widetilde{\boldsymbol{p}}=\frac{\bar{\boldsymbol{e}}}{T}}.
    \end{equation}
\end{theorem}

\begin{proof}
   See appendix \ref{app:prf-thm2}.
\end{proof}

Theorem \ref{th:capacity_eceh} essentially says that under perfect power control, a Gaussian MAC  using  energy cooperation and random energy harvesting  achieves the same capacity region as a Gaussian MAC with energy cooperation and powered by traditional battery supplies.
   In addition, the capacity region relies on the energy harvesting rate vector $\bar{\boldsymbol{e}}$ and is irrelevant to the fluctuations of the energy harvesting process.

In the proof of Theorem \ref{th:capacity_eceh}, we have generalized the analysis on point-to-point channels in \cite{Ulukus-2012-awgn} to  multiple-access channels with energy cooperation.
    Although the authors has noted that the result may be generalized to multi-user channels in \cite{Ulukus-2012-awgn},  it is not clear whether it is applicable to multi-user networks with energy cooperation, since the energy cooperation greatly complicates the analysis.
   In our model, there are $K$ independent information sources.
Also, we need to provide each user with both the energy for energy cooperation and the energy for information transmission.
   Theorem \ref{th:capacity_eceh}  is proved based on the save-and-transmit power control, where each node saves all the harvested energy in the buffer for a sufficiently long period and transmits information in the rest of the period.
When the both periods go to infinity, all nodes can perform energy cooperation and information transmission using a constant consumed power equal to the energy harvesting rate, i.e., $\widetilde{p}_k=\frac{\bar{e}_k}{T}$, without energy shortages.

\begin{remark}
    Theorem \ref{th:capacity_eceh} shows that the randomness of the energy harvesting process does not degrade the capacity region of the channel and the save-and-transmit power control is a capacity region achieving policy.
\end{remark}

\begin{remark}
    Since each user can use a constant consumed power $\widetilde{p}_k=\frac{\bar{e}_k}{T}$,  we can use the same energy cooperation policy as that for the MAC powered by traditional batteries, yet achieving the rates on the capacity region boundary.
        That is, even in the energy harvesting scenario, the capacity boundary can be achieved using a fixed energy cooperation policy.
\end{remark}

\section{Simulation}\label{sec:4}
Consider a 2-user Gaussian MAC as shown in Fig. \ref{fig:mac model}.
     We set channel gains to $h_1=0.8, h_2=1.5$ and set the energy transfer efficiencies to $\alpha_{12}=0.8, \alpha_{21}=0.5$.
For simplicity, we assume that the slot length is $T_s=1$ s, the system bandwidth is $W=1$ Hz, and the Gaussian noise at the receiver has zero mean and unit variance.
    We also assume that the energy harvested by the two users in a slot follows uniform distribution $\mathcal{U}(0,2)$ and $\mathcal{U}(0,4)$, respectively.
Thus, the energy harvesting rates of the two users are $\bar{e}_1=1$ Joule/sec and $\bar{e}_2=2$ Joule/sec, respectively.

We present the capacity region of the Gaussian MAC with both energy harvesting (EH) and energy cooperation (EC) in Fig. \ref{fig:region} (the dashed curve).
    Compared with the capacity region of a traditional Gaussian MAC using average transmit power $p_1=\bar{e}_1$ and $p_2=\bar{e}_2$ (marked by $\rhd$), it is seen that there are large capacity gains due to energy cooperations.
In fact, by using energy cooperation, energy is transferred to the user who has better channel condition and can utilize the energy more efficiently.
    Also, we see that the capacity region exactly coincides with that of a Gaussian MAC using energy cooperation and average transmit power $p_1=\bar{e}_1$ and $p_2=\bar{e}_2$  (the solid curve), which means that random energy arrival can achieve the same performance as traditional power supplies in the limitation sense.

As the transmission period $N$  increases, we investigate how fast the sum capacity of the MAC under energy harvesting constraints approaches that under average power constraints, as shown in Fig. \ref{fig:sum_c}.
    For each $N$, we run the simulation independently by generating a realization of energy harvesting process and calculating the corresponding averaged sum rate based on  save-and-transmit scheme.
We set the period of energy saving to $h(N)=N/10000$.
    On one hand, the capacity loss due to this period is negligible. On the other hand, this setting ensures the energy saving period to be large enough to eliminate the energy shortage events without requiring $N$ to be infinitely large, which makes the simulation implementable.
As is shown, when $N$ gets larger and larger, the difference between the two sum capacities become smaller and smaller, and finally vanishes.

\section{Conclusion}\label{sec:5}
We considered the capacity region $\mathcal{C}_\textrm{EC,EH}(\bar{\boldsymbol{e}})$  of Gaussian MACs with energy cooperation and energy harvesting.
    It has been proved that each rate in the capacity region is achievable using a fixed energy cooperation policy and the save-and-forward power control policy.
It is also shown that the capacity region $\mathcal{C}_\textrm{EC,EH}(\bar{\boldsymbol{e}})$  is equal to the capacity region $\mathcal{C}_\textrm{g,EC}(\widetilde{\boldsymbol{p}})$ of a Gaussian MAC with energy harvesting and average power constraint $\widetilde{\boldsymbol{p}}=\frac{\bar{\boldsymbol{e}}}{T}$.
    Based on the obtained results, one can readily  characterize the capacity region $\mathcal{C}_\textrm{EC,EH}(\bar{\boldsymbol{e}})$ explicitly through investigating $\mathcal{C}_\textrm{g,EC}(\widetilde{\boldsymbol{p}})$, which could be an interesting future direction.


\appendices
\renewcommand{\theequation}{\thesection.\arabic{equation}}

\newcounter{mytempthcnt}
\setcounter{mytempthcnt}{\value{theorem}}
\setcounter{theorem}{2}

\section{Proof of Theorem 1}\label{app:prf-thm1}
\begin{proof}
The achievability of $\mathcal{C}_{\textrm{g,EC}}(\widetilde{\boldsymbol{p}})$ is clear since it is a union of  traditional Gaussian MAC capacity regions, in which each rate  is achievable.
   A weak converse can be established by proving the following statement:
if a rate vector $\boldsymbol{R}=\{ R_1,\cdots, R_K \}$ is achievable, i.e., there exists a  $((2^{nR_1}, 2^{nR_2},\cdots,2^{nR_K}), n)$ code such that the decoding error satisfies $\lim_{n\rightarrow\infty}P^{(n)}_e\rightarrow0$, then $\boldsymbol{R}$ lies in the capacity region $\mathcal{C}_{\textrm{g,EC}}(\widetilde{\boldsymbol{p}})$  defined by (\ref{cr:ener-coop}).

For a given energy cooperation policy $\mathcal{D}$, the corresponding energy transfer matrix is $\textbf{D}=[\delta_{kj}]_{K\times K}$ and the average transmit power of user $k$ is
 \begin{equation}\label{apx:p_st}
   p_k=\widetilde{p}_k-\frac1T \sum_{l=1}^K(\delta_{kl}-\alpha_{lk}\delta_{lk}).
 \end{equation}

Denote $\mathcal{C}^n=((2^{nR_1}, 2^{nR_2},\cdots,2^{nR_K}),n)$ as a codebook with code length $n$ and code size $2^{nR_k}$ for user $k\in \Phi$, $P^{(n)}_e$ as the decoding error probability.
    Assume that the message $w_k$ of each user is drawn equiprobably from the set $\{1,2,\cdots,2^{nR_K}\}$.
The encoding function  $f_{k}^n:\{1,2,\cdots,2^{nR_K}\}\rightarrow \mathcal{X}^n$ is a mapping that assigns a codeword $\boldsymbol{x}_k=[x_{k,1},x_{k,2},\cdots,x_{k,n}]^{\textsf{T}}$ to each message $w_k$. 
    The decoding function $g^n:\mathcal{Y}^n\rightarrow \{1,2,\cdots,2^{nR_K}\}\cup\{\textsf{err}\}$ is a mapping that assigns an $\hat{w_k}$ for each user or an error message $\textsf{err}$ to each received sequence $\boldsymbol{y}_k=[y_{k,1},y_{k,2},\cdots,y_{k,n}]^{\textsf{T}}$.
The  decoding error probability is defined as
\begin{eqnarray}\nonumber
  P_e^{(n)}\hspace{-3mm}&=&\hspace{-3mm}  \frac{1}{2^{n\sum_{k=1}^K R_k}}
  \sum_{j_1=1}^{2^{nR_1}}  \cdots \sum_{j_K=1}^{2^{nR_K}}
  \Pr\{ (\hat w_{j_1},\cdots,\hat w_{j_K})\\ \nonumber
  \hspace{-3mm}&&~~~~~~~~~~~~
  \hspace{-3mm}\neq (w_{j_1},\cdots,w_{j_K})|(w_{j_1},\cdots,w_{j_K}) \textrm{~sent} \}
\end{eqnarray}

Let $\boldsymbol{R}$ be some achievable rate using codebook $\mathcal{C}^n$ so that we have $\lim_{n\rightarrow\infty}P^{(n)}_e\rightarrow0$.
    We will show that $\boldsymbol{R}$ must lie in $\mathcal{C}_{\text{g}}(\widetilde{\boldsymbol{p}})$ under some energy cooperation policy $\mathcal{D}$.
Assume that messages $\boldsymbol{w}=[w_1,w_2,\cdots,w_K]$ are sent.
    The corresponding codewords is denoted as  $[\boldsymbol{X}_1,\boldsymbol{X}_2,\cdots, \boldsymbol{X}_K]$, respectively, where $\boldsymbol{X}_k=[X_{k,1},X_{k,2},\cdots,X_{k,n}]^{\textsf{T}}$.
Denote the received signal at the receiver as $\boldsymbol{Y}=[Y_{1},Y_{2},\cdots,Y_{n}]^{\textsf{T}}$.
   By Fano's inequality,
\begin{equation}\label{apx:dr-1} \nonumber
  H(\boldsymbol{w}|\boldsymbol{Y})\leq n\boldsymbol{R}(\Phi)P_e^{(n)} + H(P_e^{(n)})\leq n\epsilon
\end{equation}
where $\boldsymbol{R}(\Phi)=\sum_{k\in \Phi} R_k$ and $\epsilon\rightarrow0$ as $P_e^{(n)}\rightarrow0$.

For any subset $S\in \Phi$, we have
\begin{equation}\label{apx:dr-2}
  H([w_k]_{k\in S}|\boldsymbol{Y})\leq H(\boldsymbol{w}|\boldsymbol{Y})\leq n\epsilon.
\end{equation}

Let us denote $Y^{i-1}=(Y_1,Y_1,\cdots,Y_{i-1})$ and $\widebar{S}=\Phi\setminus S$. Consider
  \begin{align}\nonumber
    n&\boldsymbol{R}(S)\\
     \nonumber& = H([w_k]_{k\in S})= H([w_k]_{k\in S}|[w_k]_{k\in \widebar{S}}) \\
     \nonumber & = I([w_k]_{k\in S};\boldsymbol{Y}|[w_k]_{k\in \widebar{S}}) + H([w_k]_{k\in S}|[w_k]_{k\in \widebar{S}},\boldsymbol{Y})\\
     \nonumber & = \sum_{i=1}^n I([w_k]_{k\in S};Y_i|Y^{i-1},[w_k]_{k\in \widebar{S}})+ n\epsilon\\
     \nonumber & \stackrel{(\text{a})}{=}  \sum_{i=1}^n I([w_k]_{k\in S};Y_i|Y^{i-1},[w_k]_{k\in \widebar{S}},[X_{k,i}]_{k\in \widebar{S}})+ n\epsilon\\
     \nonumber & \leq \sum_{i=1}^n I([w_k]_{k\in S},[w_k]_{k\in \widebar{S}},Y^{i-1};Y_i[X_{k,i}]_{k\in \widebar{S}})+ n\epsilon\\
     \nonumber & \stackrel{(\text{b})}{=} \sum_{i=1}^n I([X_{k,i}]_{k\in S},[w_k]_{k\in S},[w_k]_{k\in \widebar{S}},  \\
     \nonumber &~~~~~~~~~~~~~~~~~~~~~~~~~
      Y^{i-1};Y_i|[X_{k,i}]_{k\in \widebar{S}})+ n\epsilon\\
     \nonumber & = \sum_{i=1}^n I([X_{k,i}]_{k\in S};Y_i|[X_{k,i}]_{k\in \widebar{S}})+ I([w_k]_{k\in S},[w_k]_{k\in \widebar{S}},  \\
     \nonumber &~~~~~~~~~~~~~~~~~~~~~~~~~
      Y^{i-1};Y_i|[X_{k,i}]_{k\in S},[X_{k,i}]_{k\in \widebar{S}})+ n\epsilon\\
      & \stackrel{(\text{c})}{=} \sum_{i=1}^n I([X_{k,i}]_{k\in S};Y_i|[X_{k,i}]_{k\in \widebar{S}}) + n\epsilon,\label{apx:dr-3}
  \end{align}
where
 (a) and (b) hold true because $[X_k]_{k\in \widebar{S}}$ and $[X_k]_{k\in S}$ are functions of $[w_k]_{k\in \widebar{S}}$ and $[w_k]_{k\in S}$, respectively,
 and (c) follows the memoryless property of the Markov chain $([w_k]_{k\in \widebar{S}}, [w_k]_{k\in S}, Y^{i-1})\rightarrow ([X_{k,i}]_{k\in \widebar{S}},[X_{k,i}]_{k\in {S}})\rightarrow Y_i$.

Therefore, we have
\begin{equation}\label{apx:dr-4}
  \boldsymbol{R}(S)=\frac1n \sum_{i=1}^n I \left( [X_{k,i}]_{k\in S};Y_i|[X_{k,i}]_{k\in \widebar{S}} \right) + \epsilon,
\end{equation}
which is a sum of average mutual information based on the empirical distributions in the $i$-th column of the codebook.

Denote the average power of the $i$-th column of the codebook for user $k$ as $P_{k,i}$.
Since $X_{k,i}=x_{k,i}(w_k)$ and $w_k$ is uniformly distributed in $\{1,2,\cdots,2^{nR_k}\}$, we have
\begin{equation}\label{apx:dr-5}
  P_{k,i}=\frac{1}{2^{nR_k}} \sum_{w_k\in \mathcal{C}^n} x_{k,i}^2(w_{k}),
\end{equation}
which satisfies the energy constraint (\ref{apx:p_st}) as $n$ goes to infinity. That is,
\begin{equation}\label{apx:dr-7}
  \frac1n \sum_{i=1}^n P_{k,i}\leq p_k=\widetilde{p}_k-\frac1T \sum_{l=1}^K(\delta_{kl}-\alpha_{lk}\delta_{lk})
\end{equation}
for each $k\in \Phi$ and $w_k\in \mathcal{C}^n$.

In the normalized channel model, the signal at the receiver is $Y_i=\sum_{k\in S}X_{k,i}+\sum_{k\in \widebar{S}}X_{k,i}+Z_i$, where $X_{k,i}$ and $Z_i$ are independent from each other.
   Thus, the power of $Y_i$ is $\sum_{k\in \Phi}P_{k,i} +1$. Accordingly, the information rate is
\begin{equation}\label{apx:dr-6} \nonumber
\begin{split}
  \boldsymbol{R}(S) &\leq \frac1n \sum_{i=1}^n \left(h(Y_i|[X_{k,i}]_{k\in \widebar{S}}) - h(Z_i)\right) + n\epsilon\\
    & \leq \frac1n \sum_{i=1}^n \frac12 \log\left( 1+ \sum_{k\in S} P_{k,i} \right)  + n\epsilon\\
    & \stackrel{(\text{a})}{\leq} \frac12 \log\left( 1+ \frac1n  \sum_{i=1}^n\sum_{k\in S} P_{k,i} \right) + n\epsilon\\
    & = \frac12 \log\left( 1+ \sum_{k\in S} \left(\widetilde{p}_k-\frac1T \sum_{l=1}^K(\delta_{kl}-\alpha_{lk}\delta_{lk}) \right) \right) + n\epsilon,
\end{split}
\end{equation}
where (a) follows   Jensen's inequality.

It is readily seen that $\boldsymbol{R}\in\mathcal{C}_{\textrm{g,EC}}(\widetilde{\boldsymbol{p}})$, which proves the converse, and thus  Theorem \ref{th:cr-ener-coop}.
\end{proof}

\section{Proof of Theorem \ref{th:capacity_eceh}}\label{app:prf-thm2}
\begin{proof}
Denote $\mathcal{C}_\textrm{g,EC}(\widetilde{\boldsymbol{p}})|_{\widetilde{\boldsymbol{p}}=\frac{\bar{\boldsymbol{e}}}{T}}$ as the capacity region of a $K$-user Gaussian MAC with energy cooperation and average power constraint ${\widetilde{\boldsymbol{p}}=\frac{\bar{\boldsymbol{e}}}{T}}$.
    Also, we denote $\mathcal{C}_\textrm{EC,EH}(\bar{\boldsymbol{e}}) $ as the capacity region of a $K$-user Gaussian MAC with both energy cooperation and energy harvesting.
To prove Theorem \ref{th:capacity_eceh}, we first show that $\mathcal{C}_\textrm{g,EC}(\widetilde{\boldsymbol{p}})|_{\widetilde{\boldsymbol{p}}=\frac{\bar{\boldsymbol{e}}}{T}}$
is an outer bound of $\mathcal{C}_\textrm{EC,EH}(\bar{\boldsymbol{e}}) $.
   Then, the proof is completed by showing that, by using a save-forward power control \cite{Ulukus-2012-awgn}, $\mathcal{C}_\textrm{g,EC}(\widetilde{\boldsymbol{p}})|_{\widetilde{\boldsymbol{p}}=\frac{\bar{\boldsymbol{e}}}{T}}$ is in fact achievable over the $K$-user Gaussian MAC with both energy cooperation and energy harvesting.

\subsection{Capacity Region Outer Bound} \label{subsec:apx_21}
Instead of standard converse argument, we present an insightful reasoning to show that $\mathcal{C}_\textrm{g,EC}(\widetilde{\boldsymbol{p}})|_{\widetilde{\boldsymbol{p}}=\frac{\bar{\boldsymbol{e}}}{T}}$ is an outer bound of the capacity region $\mathcal{C}_\textrm{EC,EH}(\bar{\boldsymbol{e}}) $.

   Let $\boldsymbol{R}$ be an achievable rate in $\mathcal{C}_\textrm{EC,EH}(\bar{\boldsymbol{e}})$ and $\boldsymbol{X}_k=(X_{k,1},\cdots,X_{k,n})$ be a codeword of user $k$.
It is clear that $\boldsymbol{X}_k$ satisfies constraint CSTR1 so that  we have $\frac 1N \sum_{i=1}^N X_{k,i}^2\leq \frac 1N \sum_{i=1}^N e_k(i)= \frac1T \bar{e}_k$.
    This means that $\boldsymbol{X}_k$  also satisfies the average power constraint.
   Therefore,  $\boldsymbol{R}$ is also achievable over a $K$-user Gaussian MAC with energy cooperation and average power constraints.

On the contrary, although a codeword $\boldsymbol{X}_k$ designed for a $K$-user Gaussian MAC with energy cooperation and average power constraints satisfies $\frac 1N \sum_{i=1}^N X_{k,i}^2 \leq \frac1T \bar{e}_k$, it does not necessarily satisfy the energy causality constraint.
   Therefore, the capacity region of a Gaussian MAC with energy cooperation and energy harvesting is bounded by the capacity region of the corresponding Gaussian MAC with energy cooperation and average power constraints, i.e.,  $\mathcal{C}_\textrm{EC,EH}(\bar{\boldsymbol{e}}) \subseteq \mathcal{C}_\textrm{g,EC}(\widetilde{\boldsymbol{p}})|_{\widetilde{\boldsymbol{p}}=\frac{\bar{\boldsymbol{e}}}{T}}$.

\subsection{Achievability}\label{subsec:apx_22}
Suppose $\boldsymbol{R}=(R_1,R_2,\cdots,R_K)$ is an achievable rate over a $K$-user Gaussian MAC with energy cooperation policy $\mathcal{D}^*$ and average power constraints $\widetilde{\boldsymbol{p}}=\frac{\bar{\boldsymbol{e}}}{T}$.
   We denote the corresponding energy transfer matrix as $\textbf{D}^*=[\delta_{kl}]_{K\times K}$ and the transmit power as $p_k=\widetilde{p}_k-\frac1T \sum_{l=1}^K(\delta_{kl}-\alpha_{lk}\delta_{lk}) $.
Next, we shall prove that $\boldsymbol{R}$ is also achievable in the $K$-user Gaussian MAC with energy cooperation and energy harvesting, where the  energy harvesting rate is $\bar{\boldsymbol{e}}=[e_1,\cdots,e_K]$.  

We consider $N$ slots of information transmission, where each slot consists $m$ symbols.
    In the save-and-transmit scheme \cite{Ulukus-2012-awgn}, information transmission is performed in two phases: the energy saving phase and the information transmission phase.
In the energy saving phase of $h(N)\in o(N)$ slots, all the harvested energy is stored in the energy buffer and we set the consumed power $\widetilde {p}_k^\textrm{avg}$ to zero. Thus, no information is transmitted.
    In the information transmission phase of $N-h(N)$ slots, we set the consumed power of user $k$ as $\widetilde {p}_k^\textrm{avg}=\widetilde{p}_k-\varepsilon$, where $\varepsilon$ is an arbitrarily small positive number.
In particular, $h(N)$ is chosen such that both $h(N)$ and $N-h(N)$ go to infinity as $N\rightarrow\infty$.
   Under energy cooperation policy $\mathcal{D}^*$, the transmit power of user $k$ would be $p_k^\textrm{avg}=\widetilde{p}_k^\textrm{avg}-\frac1T \sum_{l=1}^K(\delta_{kl}-\alpha_{lk}\delta_{lk}) = p_k-\varepsilon$.

Note that although there are $mN$ symbols in each codeword, the size of the message set is only $2^{m(N-h(N))R_k}$.
        We assume that each message appears with equal probability and denote $\boldsymbol{X}_{k}=(X_{k,1},X_{k,2},\cdots,X_{k,mN})$ as the codeword of user $k$.
    In the save-and-transmit scheme, we have $X_{k,i}=0$ for $1\leq i \leq mh(N)$.
For $mh(N)+1 \leq i \leq N$, the symbols are selected as  independent samples  of a Gaussian distribution with zero mean and variance $\mathbb{E}[X_{k,l}^2]=p_k^\textrm{avg}$.
   Using successive decoding, it is known that the decoding error $\epsilon_1$  goes to zero as $N$ goes to infinity.
As a result, user $k$ can transmit $\log(2^{m(N-h(N))R_k})=m(N-h(N))R_k$ nats by $mN$ symbols.
   Thus, the overall data rate of user $k$ is $r_k=\frac{m(N-h(N))R_k}{mN}$, which approaches $R_k$ as $N$ goes to infinity.

Based on these analysis, the achievability can be proved by showing that the totally harvested energy in $N$ slots is sufficient for the transmission phase.
   That is, given $\widetilde{p}_k^\textrm{avg}=\widetilde{p}_k-\varepsilon$ and $\epsilon>0$, we need to show
\begin{equation}\label{apx:dr3-1}
  \Pr\left( \bigcup_{k=1}^K \bigcup_{i=1}^N A_{k,i} \right)\leq \epsilon,
\end{equation}
  where $A_{k,i}$ is the event that in slot $i$, the available energy of user $k$ is less than the required energy to transmit the first $mi$ symbols of its codeword and to perform energy cooperation.
    That is,
 \begin{eqnarray}\label{apx:dr3-2}\nonumber
   A_{k,i} \hspace{-3mm} &=& \hspace{-3mm} \Bigg\{ \hspace{-1mm} \sum_{j=1}^i \frac{e_k[j]}{T} <   \\ \nonumber
            &&  \hspace{-1mm}
            \sum_{j=1}^i \hspace{-0.5mm} \left( \hspace{-0.5mm}  \frac1m \hspace{-0.5mm}  \sum_{l_1=1}^m \hspace{-0.5mm}  X^2_{k,jm+l_1}  \hspace{-0.5mm}  + \hspace{-0.5mm}
                   \frac1T \hspace{-0.5mm}  \sum_{l_2=1}^K (\delta_{kl_2} \hspace{-0.5mm}  - \hspace{-0.5mm}  \alpha_{l_2k}\delta_{l_2k}) \right) \hspace{-1mm}  \Bigg\}.
 \end{eqnarray}

Since users do not transmit signal in the energy saving phase, we have
\begin{equation}\label{apx:dr3-3}
  \Pr\left( \bigcup_{k=1}^K \bigcup_{i=1}^{h(N)} A_{k,i} \right)=0.
\end{equation}

We denote $e_k^\textrm{sv}= \sum_{j=1}^{h(N)}e_k(j)$ as the totaly saved energy in the energy saving phase.
    For each slot in the information transmission phase, i.e.,  $h(N)+1 \leq j \leq N$, we denote the difference between the available power and the required power to transmit a symbol and to perform energy cooperation as
$$\Delta_k(j)= \frac{e_k(j)}{T} - \left( \frac1m \sum_{l_1=1}^m X^2_{k,j m+l_1} \hspace{-1mm} +\hspace{-1mm} \frac1T \sum_{l_2=1}^K (\delta_{kl_2-\alpha_{l_2k}\delta_{l_2k}}) \right).$$

Since $\mathbb{E}[X_{k,j}^2]=p_k^\textrm{avg}=\widetilde{p}_k^\textrm{avg}- \frac1T \sum_{l_2=1}^K (\delta_{kl_2}-\alpha_{l_2k}\delta_{l_2k})$ and $ \mathbb{E}[\frac{e_k(j)}{T} ]=\widetilde{p}_k$,
we have $\mathbb{E}[\Delta_k(j)] =\widetilde{p}_k-\widetilde{ p}_k^\textrm{avg}$.
   By the weak law of large numbers \cite{PRP-2001}, there exists a sufficiently large $j_0$ such that,
\begin{equation}\nonumber
   \begin{split}
      &\Pr\left\{\bigcup_{i= h(N)+j_0}^N A_{k,i} \right\} \\
      & =   \Pr \left\{ \bigcup_{ j= h(N)+j_0}^N
            \left( \frac{e_k^\textrm{sv}}{T} +  \sum_{j=h(N)+1}^i \Delta_k(j)  <0  \right)\right\} \\
      &\stackrel{(\text{a})}{<} \Pr\left\{  \bigcup_{i= h(N)+j_0}^N \left( \sum_{j=h(N)+1}^i \Delta_k(j) <0 \right) \right\} \\
       &<\Pr\Bigg\{  \bigcup_{i= h(N)+j_0}^N  \Bigg(  \Bigg| \sum_{j=h(N)+1}^i  \frac{\Delta_k(j)}{i-h(N)-1}  -   \\
                                &~~~~~~~~~~~~~~~~~~~~~~~~~~~~~~
                                    (\widetilde{p}_k -  \widetilde{ p}_k^\textrm{avg}) \Bigg| >
                                    \widetilde{p}_k - \widetilde{ p}_k^\textrm{avg}  \Bigg) \Bigg\}
                                <\epsilon',
   \end{split}
\end{equation}
where (a) is because $e_k^\textrm{sv}>0$.

Denote $c_0=\sum_{j=h(N)+1}^{j=h(N)+j_0} |\Delta_{k}(j)|$.
  Since $j_0$ is a fixed number, it is clear that $\lim_{N\rightarrow\infty}\frac{c_0}{h(N)}=0$.
For $h(N)+1\leq i \leq  h(N)+j_0$, we have
\begin{equation}\nonumber
  \begin{split}
      &\Pr\left\{ A_{k,i}\right\}=\Pr\left\{ \frac{ e_k^\textrm{sv}}{T} +  \sum_{j=h(N)+1}^i \Delta_k(j) <0 \right\} \\
                                &\stackrel{(\text{a})}{<}\Pr\left\{  \left|\frac{ e_k^\textrm{sv}}{Th(N)}-\widetilde{p}_k \right| > \widetilde{p}_k  +
                                            \frac{\sum_{j=h(N)+1}^i \Delta_k(j)}{h(N) } \right\} \\
                                &<\epsilon'',
   \end{split}
\end{equation}
where (a) follows the weak law of large numbers.

By the union bound, we have
\begin{equation}\nonumber
\begin{split}
  & \Pr\Bigg(  \bigcup_{k=1}^K \bigcup_{i=1}^N A_{k,i} \Bigg) \\
            &\leq \Pr\left( \bigcup_{k=1}^K \bigcup_{i=1}^{h(N)} A_{k,i} \right) +
                    \sum_{k=1}^K \sum_{i=h(N)+1}^{h(N)+j_0} \Pr\left( A_{k,i} \right)+   \\
             &  ~~~~
             \sum_{k=1}^K\Pr\left( \bigcup_{i=h(N)+j_0}^{N} A_{k,i} \right)\\
             & < 0 + Kj_0 \epsilon'' + \epsilon'\\
             & \triangleq  \epsilon,
\end{split}
\end{equation}
which proves the achievability of the capacity region.

It is clear that Theorem \ref{th:capacity_eceh} is proved by combining the proofs in Appendix \ref{subsec:apx_21} and \ref{subsec:apx_22}.
\end{proof}

\small{

}

\begin{biography}[{\includegraphics[width=1in,height=1.25in,clip,keepaspectratio]
{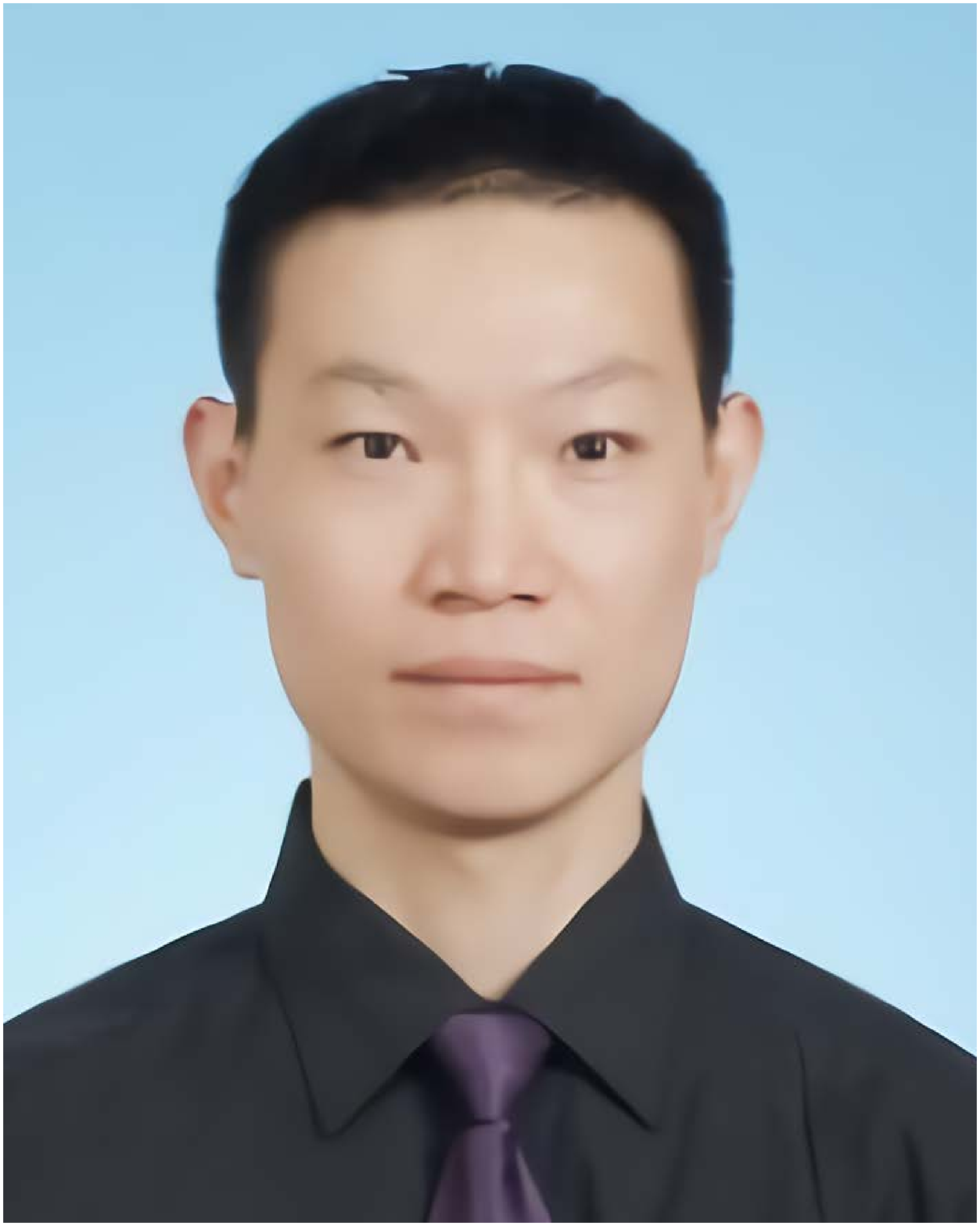}}]{Yunquan Dong} (M'15) received the M.S. degree in communication and information systems from the Beijing University of Posts and Telecommunications,
Beijing, China, in 2008, and the Ph.D. degree in communication and information engineering from Tsinghua University, Beijing, in 2014.
    He was a BK Assistant Professor with the Department of Electrical and Computer Engineering, Seoul National University, Seoul, South Korea.
He is currently a Professor with the School of Electronic and Information Engineering, Nanjing University of Information Science and Technology, China.

His research interests include heterogeneous cellular networks and energy harvesting communication systems.
    He was a recipient of the Best Paper Award of the IEEE ICCT in 2011, the National Scholarship for Postgraduates from China¡¯s Ministry of Education in 2013, the Outstanding Graduate Award of Beijing with honors in 2014, and the Young Star of Information Theory Award from the China¡¯s Information Theory Society in 2014.
\end{biography}

\begin{IEEEbiography}
[{\includegraphics[width=1in,height=1.34in]{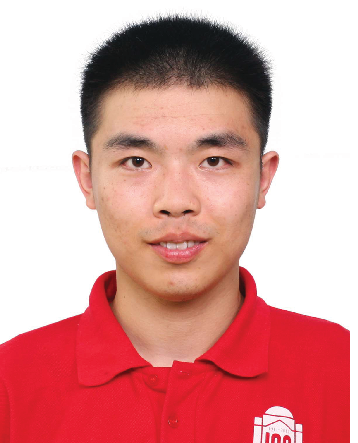}}]{Zhengchuan
Chen}(M'16) received the B.S. degree from Nankai University, Tianjin, China, in 2010 and the Ph.D. degree from Tsinghua University, Beijing, China, in 2015. From September 2012 to January 2013, he visited the Institute of Network Coding, The Chinese University of Hong Kong, Hong Kong, as a Research Assistant. From November 2013 to December 2014, he visited the Department of Electrical and Computer Engineering, University of Florida, Gainesville, FL, USA, as a Research Scholar. He is currently a Postdoctoral Research Fellow with the Information Systems Technology and Design Pillar, Singapore University of Technology and Design (SUTD), Singapore. His main research interests include wireless cooperative networks, energy harvesting, and network information theory.

Dr. Chen has served several IEEE conferences, e.g., the IEEE Global Communications Conference (Globecom); the International Symposium on Personal, Indoor, and Mobile Radio Communications (PIMRC); and the International Conference on Communications in China, as a Technical Program Committee Member. He is a reviewer for several journals of the IEEE Communications Society and was selected as an Exemplary Reviewer of the IEEE Transactions on Communications in 2015. He has received the National Scholarship for both undergraduates and postgraduates of China. He coreceived the Best Paper Award at the International Workshop on High Mobility Wireless Communications in 2013.
\end{IEEEbiography}

\begin{biography}[{\includegraphics[width=1in,height=1.25in,clip,keepaspectratio]
{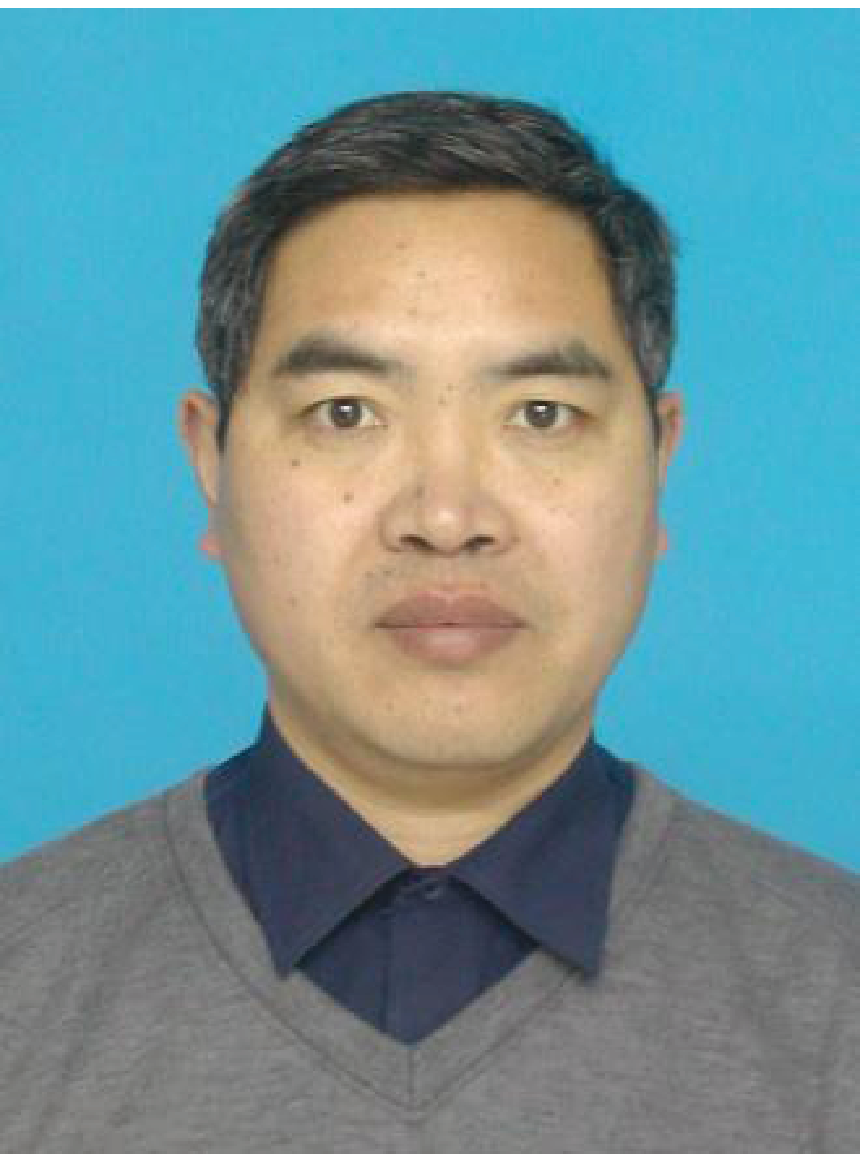}}]{Pingyi Fan} (M¡¯03-SM¡¯09) received the B.S and M.S. degrees from the Department of Mathematics of Hebei University in 1985 and Nankai University
in 1990, respectively, received his Ph.D degree from the Department of Electronic Engineering, Tsinghua University, Beijing, China in 1994.
    He is a professor of department of EE of Tsinghua University currently.
From Aug. 1997 to March. 1998, he visited Hong Kong University of Science and Technology as Research Associate. From May. 1998 to Oct. 1999, he visited University of Delaware, USA, as research fellow. In March. 2005, he visited NICT of Japan as visiting Professor. From June. 2005 to May 2014, he visited Hong Kong University of Science and Technology for many times. From July 2011 to Sept. 2011, he is a visiting professor of Institute of Network Coding, Chinese University of Hong Kong.

    Dr. Fan is a senior member of IEEE and an oversea member of IEICE.
He has attended to organize many international conferences including as General co-Chair of IEEE VTS HMWC2014, TPC co-Chair of IEEE International Conference on Wireless Communications, Networking and Information Security (WCNIS 2010) and TPC member of IEEE ICC, Globecom, WCNC, VTC, Inforcom etc. He has served as an editor of IEEE Transactions on Wireless Communications, Inderscience International Journal of Ad Hoc and Ubiquitous Computing and Wiley Journal of Wireless Communication and Mobile Computing. He is also a reviewer of more than 30 international Journals including 20 IEEE Journals and 8 EURASIP Journals.
He has received some academic awards, including the IEEE Globecom¡¯14 Best Paper Award, IEEE WCNC¡¯08 Best Paper Award, ACM IWCMC¡¯10 Best Paper Award and IEEE ComSoc Excellent Editor Award for IEEE Transactions on Wireless Communications in 2009.
     His main research interests include B5G technology in wireless communications such as MIMO, OFDMA, Network Coding, Network Information Theory and Big Data Analysis etc.
\end{biography}


\begin{thebibliography}{11}

\bibitem{WSN-eh-2011}
S. Sudevalayam and P. Kulkarni, ``Energy harvesting sensor nodes: survey and implications," \textit{IEEE Commun. Surveys Tuts.}, vol. 13, no. 3, pp. 443--461, Mar. 2011.


 \bibitem{Coop_ISIT_2012}
 B. Gurakan, O. Ozel, J. Yang, and S. Ulukus, ``Energy cooperation in energy harvesting wireless communications," in \textit{Proc. IEEE Int. Symp. Inf. Theory (ISIT'12)}, Cambridge, MA, USA, July 2012, pp. 965--969.

 \bibitem{Coop_Asilomar_2012}
B. Gurakan, O. Ozel, J. Yang, and S. Ulukus, ``Two-way and multiple-access energy harvesting systems with energy cooperation," in \textit{Proc. 46-th Asilomar Conference on Signals, Systems and Computers (ASILOMAR'12)},  Pacific Grove, CA, USA, November 2012. pp. 58--62

\bibitem{Coop_Tcom_2013}
 B. Gurakan, O. Ozel, J. Yang, and S. Ulukus, ``Energy cooperation in energy harvesting communications," \textit{IEEE Trans. Commun.}, vol. 61, no. 12, pp. 4884--4898, Dec. 2013.

 \bibitem{Zhi-2016}
Z. Chen, Y. Dong, P. Fan, and K. B. Letaief, ``Optimal throughput for two-way relaying: energy harvesting and energy co-operation",
\textit{IEEE J. Sel. Areas Commun.}, vol. 34, no. 5, pp. 1448--1462, May 2016.

\bibitem{DF_JSAC_2011}
O. Ozel, K. Tutuncuoglu, J. Yang, S. Ulukus, and A. Yener, ``Transmission with energy harvesting nodes in fading wireless channels: Optimal policies," \textit{IEEE J. Sel. Areas Commun.}, vol. 29, no. 8, pp. 1732--1743,  Aug. 2011.

%

\bibitem{MAC_JCN_2012}
J. Yang and S. Ulukus, ``Optimal packet scheduling in a multiple access channel with energy harvesting transmitters," \textit{J. Commun. Netw.}, vol. 14, no. 2, pp. 140--150, Apr. 2012.

\bibitem{IC_JCN_2012}
K. Tutuncuoglu and A. Yener, ``Sum-rate optimal power policies for wnergy harvesting transmitters in an interference channel,"  \textit{J. Commun. Netw.}, vol. 14, no. 2, pp. 151--161, Apr. 2012.

\bibitem{Xiong-2015}
K. Xiong, P. Fan, C. Zhang, and K. B. Letaief, ``Wireless information and energy transfer for two-hop non-regenerative MIMO-OFDM relay networks",  \textit{IEEE J. Sel. Areas Commun.}, vol. 33, no. 8, pp. 1595--1611, Aug. 2015.

\bibitem{Xiong-2016}
X. Di, K. Xiong, P. Fan, and H. C. Yang,  ``Simultaneous wireless information and power transfer in cooperative relay networks with rateless codes," \textit{IEEE Trans. Veh. Technol.}, online published, July 2016.

\bibitem{EC_TWC_2015}
K. Tutuncuoglu and A. Yener,    ``Energy harvesting networks with energy cooperation: Procrastinating policies," \textit{IEEE Trans. Commun.}, vol. 63, no. 11, pp. 4525--4538, Nov. 2015.


\bibitem{EC_ITA_2013}
K. Tutuncuoglu and A. Yener, ``Multiple access and two-way channels with energy harvesting and bidirectional energy cooperation," in \textit{Proc. Inf. Theory Applic. Wkshps.}, (\textit{ITA}'13),  San Diego, CA, USA, Feb. 2013,  pp. 1--8.

\bibitem{Kaya_ITA_2016}
O. Kaya, N. Su, S. Ulukus, and M. Koca, ``Delay tolerant cooperation in the energy harvesting multiple access channel," in \textit{Proc. Inf. Theory Applic. Wkshps.}, (\textit{ITA}'16),  San Diego, CA, USA, Feb. 2016,  pp. 1--8.

\bibitem{Litao-2016-2}
Tao Li,  P. Fan, Z. Chen, and K. B. Letaief, ``Optimum transmission policies for energy harvesting sensor networks powered by a mobile control center",
\textit{IEEE Trans. Wireless. Commun.}, vol. 15, no. 9, pp. 6132--6145, Sep. 2016.

\bibitem{Litao-2016}
Tao Li,  P. Fan, and K. B. Letaief, ``Outage probability of energy harvesting relay-aided cooperative networks over rayleigh fading channel",
\textit{IEEE Trans. Veh. Technol.}, vol. 65, no.2, pp. 972--978, Feb. 2016.

\bibitem{Gurakan-2016}
B. Gurakan, B. Sisman, O. Kaya, and S. Ulukus,  ``Energy and data cooperation in energy harvesting multiple access channel,"  in\textit{ Proc. IEEE Wireless Commun. Netw. Conf. }, (\textit{WCNC}'16),  Doha, Qatar, April 2016, pp. 1558--2612.

\bibitem{Ulukus-2012-awgn}
O. Ozel and S. Ulukus, ``Achieving AWGN capacity under stochastic energy harvesting," \textit{IEEE Trans. on Inform. Theory}, vol. 58, no. 10, pp. 6471--6483, Oct. 2012.

\bibitem{DongIn-2015-ICeh}
D. K. Shin, W. Choi, and D. I. Kim, ``The two-user Gaussian interference channel with energy harvesting transmitters: energy cooperation and achievable rate region," \textit{IEEE Trans. Wireless Commun.}, vol 63, no. 11, pp. 4451--4564, Nov. 2015.


\bibitem{Yangjing-2014-awgn}
J. Wu and J. Yang, ``The asymptotic equivalence between sensing systems with energy harvesting and conventional energy sources," in \textit{Proc. IEEE Global Telecommun. Conf.},  (\textit{GLOBECOM}'14), Austin, Texas, USA, Dec. 2014, pp. 1753-1758.

\bibitem{Cover_1991}
 T. M. Cover and J. A. Thomas, \textit{Elements of information theory}, New York: Wiley, 1991.

\bibitem{button-2015}
Wikipedia contributors, ``Button cell," \textit{Wikipedia, The Free Encyclopedia},
https://en.wiki pedia.org/w/index.php?title=Button$\_$cell $\&$oldid=689 713099.

\bibitem{Dijkstra-1959}
E. W. Dijkstra, ``A note on two problems in connexion with graphs," \textit{Numerische Mathematik}, no. 1, pp. 269--271, Jan. 1959.

\bibitem{PRP-2001}
G. Grimmet and D. Stirzaker, \textit{Probability and random processes}, Oxford University Press, 2001.
%
%
%
%
%
%
%

\end{thebibliography}
\end{document}